\documentclass[a4paper]{article}
\usepackage[a4paper, total={6in, 9in}]{geometry}
\usepackage{bm}
\usepackage{float}
\usepackage[T1]{fontenc}
\usepackage{changepage}
\usepackage{fixltx2e}
\usepackage{booktabs}
\usepackage{xspace}
\usepackage[dvipsnames]{xcolor}
\usepackage{graphicx}
\usepackage{enumitem}
\usepackage{fontawesome}
\usepackage[osf,sc]{mathpazo}
\usepackage[scaled=0.85]{helvet}
\usepackage{array}

\usepackage{hyperref}
\hypersetup{colorlinks=true, citecolor=teal, linkcolor=teal}

\usepackage{algorithm}
\usepackage[noend]{algpseudocode}

\def\Return{\State\textbf{return}\ }

\usepackage{amsmath,amsfonts,amssymb,amsthm}
\usepackage{cleveref}

\theoremstyle{definition}
\newtheorem{definition}{Definition}
\numberwithin{definition}{section}
\newtheorem{example}{Example}[section]

\theoremstyle{plain}
\newtheorem{theorem}[definition]{Theorem}
\newtheorem{corollary}[definition]{Corollary}
\newtheorem{proposition}[definition]{Proposition}

\theoremstyle{remark}
\newtheorem{remark}{Remark}[section]

\usepackage[
citestyle=apa,
maxcitenames=10,
bibstyle=apa,
backend=biber,
isbn=false,
alldates=year,
sortcites=true,
doi=false,
eprint=true
]{biblatex}
\bibliography{paper.bib}

\renewbibmacro*{doi+url}{%
  \iftoggle{bbx:doi}
  {\printfield{doi}%
    \iffieldundef{doi}{}{}}
  {}%
  \newunit\newblock
  \iftoggle{bbx:eprint}
  {\usebibmacro{eprint}%
    \iffieldundef{eprint}{}{}}
  {}%
  \newunit\newblock
  \iftoggle{bbx:url}
  {\usebibmacro{url+urldate}%
    \iffieldundef{url}{}{}}
  {}}

\AtEveryBibitem{\clearfield{month}}
\AtEveryBibitem{\clearfield{day}}
\DeclareNameAlias{sortname}{first-last}

\def\QQ{\mathbb{Q}}
\def\NN{\mathbb{N}}

\def\dec{\mathcal{D}}

\def\field{\mathbb{K}}
\def\acfield{\overline{\field}}
\def\Var{\mathbf{V}}
\def\Idl{\mathbf{I}}
\def\aff{\mathbb{A}}

\newcommand{\idl}[1]{\langle #1 \rangle}
\newcommand{\sat}[2]{\left(#1:#2^{\infty}\right)}
\newcommand{\rad}[1]{\sqrt{#1}}
\newcommand{\set}[1]{\{#1\}}
\newcommand{\setdes}[2]{\left\{#1\;\middle|\;#2\right\}}
\newcommand{\zclo}[1]{\overline{#1}}
\newcommand{\zeq}{\stackrel{\mathrm{\footnotesize clo}}{=}}
\DeclareMathOperator{\codim}{codim}

\newcommand{\Ffour}{\textsc{\texorpdfstring{F\textsubscript{4}}{F4}}\xspace}

\newcommand{\algname}[1]{\textsc{#1}}
\newcommand{\satb}{\algname{sat}}
\DeclareMathOperator{\hll}{hull}
\DeclareMathOperator{\irred}{Irred}

\newcommand{\sftware}[1]{\texttt{#1}\xspace}
\def\AS{\sftware{AlgebraicSolving.jl}}
\def\julia{\sftware{Julia}}
\def\msolve{\sftware{msolve}}
\def\maple{\sftware{maple}}

\def\elim{\textdagger\xspace}
\def\homolog{\textdaggerdbl\xspace}
\def\regchain{\textdollar\xspace}
\def\primedec{\textsterling\xspace}

\title{A Syzygial Method for Equidimensional Decomposition}

\author{Rafael Mohr \\{\small RPTU Kaiserslautern-Landau, Germany and
    Sorbonne Université, LIP6, CNRS, Paris, France}} \date{}

\begin{document}

\maketitle

\begin{abstract}
  Based on a theorem by Vasconcelos, we give an algorithm for
  equidimensional decomposition of algebraic sets using syzygy
  computations via Gröbner bases. This algorithm avoids the use of
  elimination, homological algebra and processing the input equations
  one-by-one present in previous algorithms. We experimentally
  demonstrate the practical interest of our algorithm compared to the
  state of the art.
\end{abstract}

\section{Introduction}
\label{sec:int}

\paragraph*{Problem Statement \& Motivation}

Let $\field$ be a field with algebraic closure $\acfield$ and let $R$
be a polynomial ring over $\field$ in $n$ variables. Recall that an
algebraic set (for us the set of common zeros of a set of polynomials
in $R$) in the affine space $\aff^n$ over $\acfield$ may be written
uniquely, up to reordering, as a finite union of $\field$-irreducible
algebraic sets, which are those defined by prime ideals in $R$. We say
that an algebraic set is {\em equidimensional} if all of its
$\field$-irreducible components have the same dimension.

Given a finite sequence $F\subset R$ denote by
$\Var(F)\subset \aff^n$ the corresponding algebraic set. Our goal is to
give an algorithm which produces finite sets of polynomials
$G_1,\dots,G_s$ s.t. each of the algebraic sets $Y_i:=\Var(G_i)$ is
equidimensional and s.t.
\[\Var(F) = \bigcup_{i=1}^s Y_i.\]

Algorithms for equidimensional decomposition frequently find
application in real algebraic geometry
\parencite[e.g.][]{safeyeldin2004, prebet2024} where certain critical
point computations require equidimensionality on their input. Somewhat
related, we also mention our previous work \textcite{helmer2024a},
where Whitney stratifications are computed for equidimensional
input. Other applications include automated theorem proving in
geometry \parencite[e.g.][]{wu2007, yang2001, chen2013} and the design
of vision-based controllers in robotics
\parencite[e.g.][]{garciafontan2022, pascual-escudero2021}.

\paragraph*{State of the Art}

In the world of symbolic computation, a vast collection of algorithms
exists to compute equidimensional decompositions, distinguished by the
data structure they use to encode algebraic sets.

The first family of algorithms is based on a data structure called a
{\em geometric resolution}. Such algorithms have been developed in
\textcite{lecerf2000, giusti2001, lecerf2003}. Roughly speaking, a
geometric resolution describes an algebraic set $X\subset \aff^n$ as
follows: Recall that if $X$ is zero-dimensional then it may be
parametrized over the zeros of a single univariate polynomial, in this
case a geometric resolution is precisely this parametrization. If now
$d:=\dim(X) > 0$ and $X$ is equidimensional, we choose a subspace
$\aff^d\subset \aff^n$ with $d = \dim(X)$ w.r.t. to which $X$ is in
Noether position, i.e. the projection from $X$ to $\aff^d$ has
everywhere zero-dimensional fiber. We may then describe the generic
fiber of this projection with a geometric resolution as in the
zero-dimensional case with base field the function field corresponding
to $\acfield^d$ instead of $\field$. \textcite{schost2003} describes
an algorithm which avoids the Noether position assumption (i.e. the
subspace $\aff^d$ is fixed) but requires the input to be given
by $n-d$ equations.

Algorithms based on geometric resolutions utilize so-called {\em
  straight line programs} for fast, division-free evaluation of
polynomials and a formal Newton iteration algorithm, which relies on
assuring certain non-degeneracy assumptions. These algorithms obtain
the best known complexity bounds for equidimensional decomposition,
polynomial in a quantity derived from the degree of the algebraic set
cut out by the input system. As far as we know, there exist no
implementations of algorithms utilizing geometric resolutions to
compute equidimensional decompositions in currently used computer
algebra systems.

We also mention \textcite{jeronimo2002} for an approach using similar
techniques.

Another family of algorithms uses {\em triangular sets}. The basic
idea of triangular sets is that each equidimensional algebraic set of
codimension $c$ should be cut out, at a generic point, by exactly $c$
equations. If we enforce these $c$ equations to have in addition a
triangular structure w.r.t. the variables of the underlying polynomial
ring we obtain a triangular set. Each polynomial in a given triangular
set has then attached to it a {\em leading variable}. A triangular set
again describes an associated algebraic set by describing its
(zero-dimensional) generic fiber under the projection to the subspace
corresponding to those variables which do not appear as a leading
variable. This generic fiber is cut out by the constituent polynomials
of the triangular set. Note that not every equidimensional algebraic
set can be described by such a triangular set (such algebraic sets are
frequently called {\em characterizable} or {\em equiprojectable}), but
every algebraic set may be written as a union of characterizable
algebraic sets.

Such triangular sets have their origin in so-called Wu-Ritt
characteristic sets \parencite[see e.g.][]{ritt1950a, wu1986,
  chou1990, wang1993, Gallo1991} which are prominently used in
differential algebra.

Algorithmically, the prescribed triangular structure yields suitable
generalizations of algorithms for univariate polynomials to work with
$X$, in particular Euclid's algorithm, through the so-called
D5-principle \parencite{delladora1985}.

Of particular importance, especially in the realm of equidimensional
decomposition, are certain special triangular sets called
\textit{regular chains}, introduced by \textcite{kalkbrener1993,
  lu1994} in which algorithms are given to decompose a given algebraic
set into ones described by regular chain. A regular chain is a
triangular set $T$ where in addition each solution to the first $k$
equations of $T$ (sorted by leading variable) can be lifted to
solutions of the first $k+1$ equations. Another decomposition
algorithm using regular chains was given by \textcite{lazard1991}.

Algorithms using regular chains are prominently part of the computer
algebra system \maple \parencite{lemaire2005, chen2007a, chen2012}. We
further refer to \textcite{wang2001, hubert2003} for comprehensive
introductions to the subject.

Our algorithm is based on Gröbner bases, invented under this name by
\textcite{buchberger1965, buchberger2006}. Instead of encoding an
algebraic set $X$ on a Zariski-open subset of itself, where specific
geometric circumstances may be assumed, as is the case for geometric
resolutions and triangular sets, Gröbner bases work by effectively
describing a polynomial ideal $I$ with $\Var(I) = X$. For any such
polynomial ideal, a Gröbner basis consists of a generating set for $I$
which provides an effective ideal membership test (i.e. given $f\in R$
we may test whether $f\in I$). A Gröbner basis of $I$ makes a large
number of algebro-geometric invariants and operations for $I$ or $X$
computable (we may compute for example the dimension or degree of $I$,
or intersect $I$ with another ideal). Gröbner bases exist for any
polynomial ideal or, in fact, any submodule of a free module over
$R$. This flexibility makes them an important cornerstone of all
popular computer algebra systems. Gröbner bases depend in addition on
the choice of a {\em monomial order} (which is a certain total order
on the set of all monomials of $R$) and their properties depend, at
least partially, on the choice of this monomial order.

A number of algorithms for equidimensional decomposition based on
Gröbner bases exist. Most of them are based on elimination techniques
which they use to compute with generic fibers of projections
\parencite[see e.g.][]{decker1999, gianni1988} or they use the fact
that a suitable projection of an equidimensional algebraic set is a
hypersurface, i.e. given by a single polynomial \parencite[see
e.g.][]{caboara1997}. Another widely implemented algorithm was given
by \textcite{eisenbud1992}. They give a homological criterion for
equidimensionality which is then turned into an algorithm via the
computation of free resolutions of polynomial ideals using Gröbner
bases. Finally, with particular relation to this work, we have the
algorithms by \textcite{moroz2008} and our previous work
\parencite{eder2023, eder2023a} which work essentially by detecting
whether a hypersurface intersects a given equidimensional algebraic
set $X$ regularly (meaning this intersection has codimension exactly
one less than $X$) in order to avoid potentially costly elimination
operations or the computation of full free resolutions.

\paragraph*{Contributions}

A number of the algorithms presented in the last section proceed by
processing the defining equations of the algebraic set to be
decomposed one-by-one: Given input $F:=(f_1,\dots,f_r)\subset R$, one starts
by intersecting the equidimensional algebraic set $\Var(f_1)$ with
$\Var(f_2)$ and decomposing the result into equidimensional algebraic
sets. The resulting sets are then intersected with $\Var(f_3)$ and
decomposed again. Proceeding like this until all polynomials in $F$
have been processed, one obtains an equidimensional decomposition of
$\Var(F)$. Such a strategy is employed, for example, by all algorithms
using geometric resolutions mentioned in the last section, by Lazard's
algorithm for triangular sets \parencite{lazard1991}, or by the
algorithm of \textcite{moroz2008} and our previous work
\parencite{eder2023a} in the world of Gröbner bases.

Let us denote for an algebraic set $X$ its set of $\field$-irreducible
components by $\irred(X)$. A frequent issue with such an incremental
strategy is that it tends to produce superfluous irreducible components
in the sense that, after having decomposed $X=\bigcup_iY_i$, the inclusion
\begin{equation}
  \label{eq:irreddec}
  \irred(X)\subseteq \bigcup_i\irred(Y_i)
\end{equation}
is usually strict. Because of the incremental method employed, if
$X = \Var(f_1,\dots,f_k)$ for $k< r$, these redundant components can
then cause additional work in the next step, when the task is to
decompose $Y_i\cap \Var(f_{k+1})$ for each $i$.

When using Gröbner bases with such an incremental strategy, another
problem arises: We observed that when using our own incremental
algorithm \parencite{eder2023a}, that the associated Gröbner basis
computations tend to be hardest ``in the middle'', i.e. when about
half of the equations have been processed, in particular when this
algorithm does not produce any decomposition after having processed
the first few equations. This behaviour can be somewhat explained by
complexity statements for Gröbner basis computations under some
regularity assumptions: Indeed, for a so-called {\em regular sequence
  in strong Noether position}, the cost of computing a Gröbner basis
equation-by-equation by echelonizing Macaulay matrices is higher than
in the final steps \parencite{bardet2015}. Dimension dependent
complexity bounds provide another confirmation of this behaviour, see
\textcite{hashemi2017}.

These two problems motivate the design of further {\em
  non-incremental} algorithm for equidimensional decomposition using
Gröbner bases. Such algorithms do exist of course, the
elimination-based algorithms mentioned previously and the algorithm by
\textcite{eisenbud1992} are non-incremental. The former requires the
computation of Gröbner bases for elimination monomial orders, which
tends to be harder than computing Gröbner bases for the
degree-reverse-lexicographical (DRL) monomial order \parencite[see
e.g.][for a theoretical investigation of this
phenomenon]{lazard1983}. The latter algorithm involves heavy
computational machinery, in particular the computation of free
resolutions and kernels between modules over $R$, and it has been
observed by \textcite{decker1999} that this algorithm tends to be
slower than elimination-based techniques on larger examples.

This work gives another non-incremental algorithm based on a
non-incremental characterization of local regular sequences by
\textcite{vasconcelos1967}. To decompose $\Var(F)$, we only use
partial information about the syzygies (i.e. polynomial relations) of
$F$, thus avoiding the use of elimination and the computation of full
free resolutions while also not being forced to iterate over the
equations in $F$ one-by-one as for our previous algorithms. The
decomposition produced by our algorithm in addition is completely
irredundant, i.e. the inclusion in \Cref{eq:irreddec} turns into an
equality for our algorithm. Experimentally, we show the practical
efficiency of our algorithm in comparison with our previous algorithm
\parencite{eder2023a}, an algorithm for decomposition based on regular
chains implemented in \maple, an elimination-based algorithm using
Gröbner bases and the algorithm by \textcite{eisenbud1992}.

\paragraph*{Outline}

In \Cref{sec:algdesc} we give the necessary algebraic preliminaries,
in particular the underlying theorem by \textcite{vasconcelos1967}.
Based on this, we then describe the main ideas of our algorithm.
Then, in \Cref{sec:datstruc}, we show how to turn these ideas into a
concrete algorithm using Gröbner bases and prove the correctness and
termination of this algorithm. We do not define or introduce Gröbner
bases as their inner workings are inessential to the understanding of
the geometric ideas underlying our algorithm and how we perform the
necessary operations with them is part of the literature. Finally, in
\Cref{sec:bench}, we describe our implementation of our
algorithm. This implementation is based on so-called {\em
  signature-based Gröbner basis algorithms} \parencite[see
e.g.][]{eder2017} to facilitate the necessary syzygy computations
\parencite[following e.g.][]{sun2011} and uses the Gröbner basis library
\msolve \parencite{berthomieu2021} based on the \Ffour algorithm
\parencite{faugere1999} to perform all remaining tasks. We give some
benchmarks for this implementation, comparing it to several other
decomposition algorithms.

\paragraph*{Acknowledgements}

This work has been supported by the DFG Sonderforschungsbereich TRR
195 and the Forschungsinitiative Rheinland-Pfalz. The author thanks
Christian Eder, Pierre Lairez and Mohab Safey El Din for valuable
discussions and Benjamin Briggs for elucidation on Vasconcelos'
theorem.

\section{Algebraic Description \& Preliminaries}
\label{sec:algdesc}

In this section, we give the necessary preliminaries from commutative
algebra and algebraic geometry to then describe our algorithms in
purely algebraic terms without concern for particular data structures.

We start by introducing our notation. Let $R$ be a polynomial ring in
$n$ variables over a field $\field$. Denote by $\aff^n$ the affine
space over an algebraic closure of $\field$. For a set $F\subset R$ we
denote by $\Var(F)$ the algebraic set defined by $F$ and for a set
$X\subset \aff^n$ we denote by $\Idl(X)$ the ideal of polynomials vanishing
on $X$. If $X$ is algebraic, we denote by $\irred(X)$ the set of
$\field$-irreducible components of $X$.

For a prime ideal $P\subset R$ we denote by $R_P$ the localization of
$R$ at the multiplicatively closed subset $R\setminus P$ and for an element
$h\in R$ we denote by $R_h$ the localization of $R$ at the
multiplicatively closed set $\setdes{h^k}{k\in \NN}$, see Section 2.1 in
\textcite{eisenbud1995} for a definition of localizations.

Further recall

\begin{definition}[Locally Closed Set]
  A {\em locally closed set} in $\aff^n$ is the set difference
  of two algebraic sets.
\end{definition}

Ideal-theoretically, locally closed sets are described by {\em
  saturation ideals}:

\begin{definition}[Saturation]
  Given two ideals $I,J\subset R$, the {\em saturation ideal} of $I$ by
  $J$ is defined by
  \[\sat{I}{J}:=\setdes{f\in R}{fJ^k\subset I\text{ for some }k\in \NN}.\]
  If $J = \idl{f}$ is generated by a single element $f\in R$ we simply
  write $\sat{I}{f}$ instead of $\sat{I}{\idl{f}}$.
\end{definition}

Now recall

\begin{proposition}[see e.g. Chapter 4 in \cite{cox2015a}]
  \label{prop:loccloideal}
  For two ideals $I,J\subset R$ we have
  \[\zclo{\Var(I)\setminus \Var(J)} = \Var(\sat{I}{J}).\]
  In particular, the irreducible components of $\Var(\sat{I}{J})$ are
  those irreducible components of $\Var(I)$ not contained in
  $\Var(J)$.
\end{proposition}

Given a finite sequence $F:=(f_1,\dots,f_r)\subset R$ our goal will be 
to algorithmically decompose the algebraic set $\Var(F)$ into
{\em equidimensional} locally closed sets:

\begin{definition}[Equidimensionality]
  A locally closed set $Y\subset \aff^n$ is called {\em equidimensional}
  if all irreducible components of the Zariski closure $\zclo{Y}$
  of $Y$ have the same dimension.
\end{definition}

To do that we will use the concept of {\em local regular sequences}:

\begin{definition}[Local Regular Sequence]
  A finite sequence $F:=(f_1,\dots,f_r)\subset R$ is called a {\em local
    regular sequence} (outside of $\Var(h)$ for some $h\in R$) if for
  every prime ideal $P\subset R$ (with $h\notin P$), $F$ defines a regular
  sequence in the ring $R_P$, i.e. for each $1< i \leq r$, $f_i$ is not a
  zero divisor in the ring $R_P/\idl{f_1,\dots,f_{i-1}}R_P$.
\end{definition}

Local regular sequences give equidimensional locally closed sets:

\begin{proposition}
  Suppose $F:=(f_1,\dots,f_r)\subset R$ defines a local regular sequence outside
  of $\Var(h)$ for some $h\in R$. Then $\Var(F)\setminus \Var(h)$ is equidimensional
  of codimension $r$.
\end{proposition}
\begin{proof}
  This is standard and follows e.g. from the fact that $R_h$ is a
  Cohen-Macaulay ring, see in particular Theorem 17.3 and 17.6 in
  \textcite{matsumura1987}.
\end{proof}

We will want to detect when a given finite sequence $F\subset R$ forms a
local regular sequence. To this end we use the following Theorem by
\textcite{vasconcelos1967}, which forms the core ingredient of our
algorithm:

\begin{theorem}
  \label{thm:conorm}
  Let $f_1,\dots,f_r\in R$, $I:= \idl{f_1,\dots,f_r}$ and $h\in R$. Then
  $f_1,\dots,f_r$ is a local regular sequence outside of $\Var(h)$ if
  and only if $(I/I^2)_h$ is a free $(R/I)_h$-module, with basis given
  by the images of $f_1,\dots,f_r$ in $I/I^2$.
\end{theorem}

Made somewhat more computationally explicit, this theorem yields:

\begin{corollary}
  \label{prop:syz}
  Let $f_1,\dots,f_r\in R$ and let $h\in R$. Suppose that for every
  $g\in R$ with \[gf_i\in \idl{f_1,\dots,f_{i-1},f_{i+1},\dots,f_r}\] for
  some $i$ we have $g\in \sat{I}{h}$. Then $(I/I^2)_h$ is free of rank
  $r$ over $(R/I)_h$, freely generated by the images of
  $f_1,\dots,f_r$.
\end{corollary}
\begin{proof}
  Suppose that in $R_h$ we have a relation
  \[\sum_i g_if_i = \sum_i\sum_{j\leq i}g_{ij}f_jf_i \]
  i.e. $\sum_ig_if_i= 0$ in $(I/I^2)_h$. Then we can write
  \[\sum_i\left( g_i - \sum_{j\leq i}g_{ij}f_j\right) f_i=0.\]
  By assumption we then have
  \[(g_i - \sum_{j\leq i}g_{ij}f_j)h^k\in I\] for every $i$ and some
  $k\in \mathbb{N}$. This implies $g_i\in I_h$ for every $i$ or in other words
  $g_i=0$ in $(R/I)_h$ for every $i$. This proves that $(I/I^2)_h$ is
  freely generated by the images of $f_1,\dots,f_r$ over $(R/I)_h$.
\end{proof}

Fix now a sequence $F:=(f_1,\dots,f_r)$. We will decompose $\Var(F)$
into equidimensional locally closed sets in the following sense:

\begin{definition}[Closure Partition]
  Let $X\subset \aff^n$ be a locally closed set. For another locally closed set $Y\subset \aff^n$ we write
  \[X \zeq Y\]
  if $\zclo{X} = \zclo{Y}$.
  A {\em closure partition}
  of $X$ is a finite set of pairwise disjoint locally closed sets $\dec$
  such that
  \[X \zeq \bigcup_{Y\in \dec}Y\] and such that in addition $\irred(\zclo{Y_1})\cap \irred(\zclo{Y_2})=\emptyset$ for $Y_1\neq Y_2\in \dec$ and
    \[\irred(\zclo{X}) = \bigcup_{Y\in \dec}\irred(\zclo{Y}).\]
\end{definition}

With this definition, we will compute an {\em irredundant Kalkbrener
  partition} of $\Var(F)$:

\begin{definition}[Irredundant Kalkbrener Partition]
  \label{def:kalkpart} An {\em irredundant Kalkbrener partition} of a
  locally closed set $X\subset \aff^n$ is a finite set of locally closed
  sets $\dec$ such that $\dec$ forms a closure partition of
  $X$ and such that each $Y\in \dec$ is
  equidimensional.
\end{definition}

\begin{remark}
  The name ``Kalkbrener partition'' is in recognition of the algorithm
  by \textcite{kalkbrener1993} which computes a similar decomposition.
\end{remark}

Let $X:=\Var(F)$. To compute an irredundant Kalkbrener Partition of
$X$, we will start by identifying a polynomial $g\in R$ with
$gf_i\in \idl{f_1,\dots,f_{i-1},f_{i+1},\dots,f_r}$ and $g\notin I$. This
means that $f_1,\dots,f_r$ is not a local regular sequence. We then
compute a closure partition of $X:=\Var(F)$ into two locally closed
sets $X_1$ and $X_2$ which will satisfy
\begin{align*}
  &\zclo{X_1} = \bigcup_{\substack{Y\in\irred(X)\\g\notin \Idl(Y)}}Y,\\
  &\zclo{X_2} = \bigcup_{\substack{Y\in\irred(X)\\g\in \Idl(Y)}}Y.
\end{align*}
By \Cref{prop:loccloideal}, for $X_1$, we may simply choose
\[X_1:=\Var(f_1,\dots,f_{i-1},f_{i+1},\dots,f_r)\setminus \Var(g).\]

We will exclusively work with locally closed sets of the shape
$Y\setminus \Var(h)$ with $h\in R$ and $Y$ algebraic. Thus we want to now
represent $X_2$ via a closure partition into locally closed sets of
this shape. To this end we define:

\begin{definition}[Hull]
  Let $X$ be a locally closed set and let $p\in R$. The {\em hull} of
  $X$ and $p$ is defined as the locally closed set
  \[\hll(X,p) := X\setminus \left[\zclo{X\setminus \Var(p)}\right].\]
\end{definition}
The hull of $X$ and $g$ gives us precisely those irreducible components
of $X$ on which $g$ vanishes:
\begin{proposition}
  \label{prop:hullcomps}
  Let $X$ be a locally closed set, let $p\in R$ and let $Y:= \hll(Y,p)$.
  Then
  \[\irred(\zclo{Y}) = \bigcup_{\substack{Z\in\irred(X)\\g\in \Idl(Z)}}Z.\]
\end{proposition}
\begin{proof}
  By \Cref{prop:loccloideal} the irreducible components of
  $\zclo{X\setminus \Var(p)}$ consist of those irreducible components of
  $\zclo{X}$ on which $p$ is not identically zero. Then the
  irreducible components of $\hll(X,p)$ are precisely the complement
  of these irreducible components, i.e. the ones on which $p$ is
  identically zero.
\end{proof}

Suppose now that we have computed some finite generating set
$P\subset R$ for an ideal cutting out $\zclo{X_1}$. Choose $p\in P$ and let
$P':=P\setminus \{p\}$.  Then we compute
\begin{align*}
  X_2 := \hll(X,g) &= X\setminus \Var(P)\\
 & \zeq X\setminus \Var(p) \sqcup \hll(X\setminus \Var(P'), p).
\end{align*}
So we are decomposing
$X\setminus \Var(P)$ into the union of the components of $X$ where $p$ does
not vanish and the union of the components where $p$ vanishes but one
of the elements in $P'$ does not. This formula gives us now a
recursive algorithm to compute a closure partition of the chosen
$X_2$, which will terminate because $R$ is Noetherian.
\begin{example}
  To briefly illustrate this recursive strategy for computing
  $\hll(X,g)$ let $R:=\QQ[x,y,z]$ and $X = \Var(xy,xz,yz)$ and
  $g = y$. First note that $\sat{\idl{xy,xz,yz}}{g} =
  \idl{x,z}$. Therefore
  \begin{align*}
    \hll(X,g) &= X\setminus \Var(x,z)\\
              & \zeq X\setminus \Var(x) \sqcup \hll(X\setminus \Var(z), x)\\
              &= \Var(y,z)\setminus \Var(x) \sqcup \hll(\Var(x,y), x)\\
              &= \Var(y,z)\setminus \Var(x) \sqcup \Var(x,y).
  \end{align*}
\end{example}

Once this recursive algorithm terminates with a
list of locally closed sets $Y_1,\dots,Y_s$ of the form
\[Y_i = \Var(f_1,\dots,f_r)\setminus \Var(h_i)\] we will replace the original
sequence $f_1,\dots,f_r$ by $f_1,\dots,f_r,g$ in each of these locally
closed sets and recursively go through the same process with $X_1$ and
all the $Y_i$.

\begin{remark}
  Perhaps a more obvious idea than the algorithmic strategy lined out
  above is to, upon finding $g$ as above, split $X$ into
  \[X = \left[X\setminus \Var(g)\right] \sqcup \left[X\cap \Var(g)\right].\] Note that
  this is not a closure partition of $X$, we introduce components that
  $X$ did not have.

  For example, taking $X = \Var(xy,xz)\subset \aff^3$ as above, we have for
  $g:=y$
  \[g\cdot xz \in \idl{xy}\]
  and if we factor
  \[X = \Var(xy)\setminus \Var(x) \sqcup \Var(xy,xz,y)\] then the algebraic set
  $\Var(xy,xz,y) = \Var(y,xz)$ has the irreducible component
  $\Var(x,y)$ which $X$ did not have.

  One may accept this, because it yields a simpler algorithm, but we
  found experimentally that this behaves worse than decomposing $X$
  into a closure partition as above. The superfluous components
  introduced cause an exponential blow up and a significant slow down
  in the number of output locally closed sets produced on certain
  examples.
\end{remark}

\section{The Data Structure}
\label{sec:datstruc}

Again, let $F:=(f_1,\dots,f_r)\subset R$ be a finite sequence. To implement
the ideas from \Cref{sec:algdesc} into a concrete algorithm, we now
need a data structure to represent locally closed sets of the form
$X:=\Var(F)\setminus \Var(h)$ with $h\in R$. For this we will use Gröbner bases,
which we do not introduce here. We refer to \textcite{cox2015a} for an
introduction to the subject and to \textcite{becker1993} for a
comprehensive overview.

A Gröbner basis of a polynomial ideal is in particular a generating
set of the same ideal and given $F$ and $h$ one is able to compute a
Gröbner basis of $\sat{\idl{F}}{h}$ using the classical {\em
  Rabinowitsch trick} \parencite{rabinowitsch1930} together with the
{\em elimination theorem} \parencite[e.g. Theorem 2 in Chapter 3
of][]{cox2015a}, we denote this operation by $\satb(F,h)$. Note that
elimination methods are not strictly required to compute saturations,
we also mention \textcite{berthomieu2023} and our previous work
\textcite{eder2023}.

We represent such a locally closed set by
\begin{enumerate}
\item Storing the sequence $F$ and the polynomial $h$.
\item Storing a Gröbner basis $G$ of the ideal
  $I_X:=\sat{\idl{f_1,\dots,f_r}}{h}$.
\end{enumerate}
We denote a member of this data structure by $\Var(F,h,G)$ and call
it an {\em affine cell}.
\begin{remark}
  Recall that a Gröbner basis depends on choosing a {\em monomial order}
  on $R$, i.e. a certain total order on the set of monomials in $R$.
  For our purposes we may choose any monomial order. In practice
  we always choose the {\em degree-reverse-lexicographic} (DRL)
  order for which Gröbner bases tend to be easiest to compute
  using Buchberger's algorithm or a related algorithm, such as \Ffour.
\end{remark}

In the following, we now require the following basic operations for
affine cells:

\begin{algorithm}[H]
  \caption{Basic Operations for Affine Cells}
  \label{alg:basop}
  \begin{algorithmic}[1]
    \Require An affine cell $X:=\Var(F,h,G)$ a polynomial $g\in R$.
    \Ensure An affine cell representing the locally closed set $\Var(F)\setminus \Var(h) \cap \Var(g)$.
    \Procedure{AddEquation}{$X,g$}
      \State $G' \gets \satb(G\cup \set{g}, h)$
      \Return $\Var(F\cup \set{g}, h, G')$
    \EndProcedure

    \Require An affine cell $X:=\Var(F,h,G)$ a polynomial $g\in R$.
    \Ensure An affine cell representing the locally closed set $\Var(F)\setminus \Var(gh)$.
    \Procedure{AddInequation}{$X,g$}
      \State $G' \gets \satb(G, g)$
      \Return $\Var(F, gh, G')$
    \EndProcedure

    \Require An affine cell $X:=\Var(F,h,G)$
    \Ensure \texttt{true} if $X$ represents the empty set, \texttt{false} otherwise.
    \Procedure{IsEmpty}{$X$}
      \Return $1\in G$.
    \EndProcedure

    \Require An affine cell $X:=\Var(F,h,G)$
    \Ensure The codimension of the locally closed set defined by $X$.
    \Procedure{Codim}{$X$}
      \State $D\gets $ the dimension of $\idl{G}$, computed via Theorem 11 in Chapter 9 of \textcite{cox2015a}
      \State $n\gets $ the dimension of the ambient affine space of $X$
      \Return $n-D$
    \EndProcedure
  \end{algorithmic}
\end{algorithm}

In order to apply \Cref{prop:syz}, given an affine cell
$X:=\Var(F, h, G)$, we give ourselves in addition a black box
algorithm $\algname{GetSyz}(X)$ that returns, if it exists, a
polynomial $g\in R\setminus I_X$ and an integer $i$ with
$1\leq i\leq r$ s.t.
$gf_i\in \idl{f_1,\dots,f_{i-1},f_{i+1},\dots,f_r}$ or the tuple $0,0$
if no such $g$ exists. This boils down to two steps: First, we need to
compute a syzygy, i.e. polynomial relation, of the sequence
$F$. Recall that a syzygy of the sequence $F$ is a vector of
polynomials $(g_1,\dots,g_r)\in R^r$ s.t.
\[\sum_{i=1}^{\infty} g_if_i = 0.\]
As a second step, given such a syzygy, we then check the entries of
these syzygies for ideal membership in $I_X$ using $G$ and {\em normal
  form computations}, see e.g. Section 3 of Chapter 2 and Corollary 2
in Section 6 of Chapter 2 of \textcite{cox2015a}.

To compute these syzygies we may use either standard methods
\parencite[see e.g. Section 15.10.8 in][]{eisenbud1995} or {\em
  signature-based Gröbner basis techniques} following
e.g. \textcite{gao2010} or \textcite{sun2011}, which we use in
practice. The latter has two advantages:
\begin{enumerate}
\item Syzygies of $F$ are met naturally one by one {\em during} a
  Gröbner basis computation for $\idl{F}$, this means that once we
  meet such a syzygy we may halt our Gröbner basis computation of
  $\idl{F}$ and check the entries of the identified syzygy for
  membership in $I_X$.
\item When ran in a certain way, the syzygies compute by such a
  signature-basea algorithm are ``non-trivial'' in the sense that they
  do not lie in the submodule generated by the {\em Koszul syzygies}
  of $F$, where the latter are syzygies of the form
  $f_if_j - f_jf_i = 0$ for $i\neq j$, see e.g. Corollary 7.2 in
  \textcite{eder2017}.
\end{enumerate}

\section{The Algorithms}
\label{sec:algos}

We now present our algorithms for computing a Kalkbrener partition for
a given $\Var(F)$ with $F:=(f_1,\dots,f_r)\subset R$ a finite
sequence. Before presenting the main algorithm let us first present
our implementation of computing the hull of a locally closed set,
represented by an affine cell, and a polynomial in $R$:

\begin{algorithm}[H]
  \caption{Computing the Hull}
  \label{alg:hull}
  \begin{algorithmic}[1]
    \Require An affine cell $X$, an element $g\in R$
    \Ensure A closure partition of $\hll(X,g)$ 
    \Procedure{Hull}{$X,g$}
    \State $G\gets$ the underlying Gröbner basis of $X$
    \State $H \gets \satb(G,g)$
    \Return $\Call{Remove}{X, H}$
    \EndProcedure
  \end{algorithmic}
  \medskip
  \begin{algorithmic}[1]
    \Require An affine cell $X$, a finite set $P\subset R$
    \Ensure A closure partition of $X\setminus \Var(P)$
    \Procedure{Remove}{$X,P$}
    \State $p\gets $any element in $P$
    \State $P'\gets P\setminus \set{p}$
    \State $\dec \gets \Call{Remove}{X,P'}$
    \State $Y\gets \Call{AddInequation}{X, p}$
    \If{$I_X = I_Y$} \qquad \emph{checked via normal form computations}
    \Return $Y$
    \ElsIf{$\Call{IsEmpty}{Y}$}
    \Return $\dec$
    \Else
    \Return\label{line:crl2}$\set{Y}\cup \bigcup_{Z\in \dec}\Call{Remove}{Z,G_Y}$
    \EndIf
    \EndProcedure
  \end{algorithmic}
\end{algorithm}

\begin{theorem}
  \label{thm:hllremove}
  The procedures \algname{Hull} and \algname{Remove} terminate and
  are correct in that they satisfy their output specifications.
\end{theorem}
\begin{proof}
  We use the notation of \Cref{alg:hull} throughout.

  The correctness and termination of \algname{Hull} are, by definition
  of the hull operator, immediate once the correctness and termination
  of \algname{Remove} is established. The correctness of
  \algname{Remove} is proven once we can prove that
  \[X\setminus \Var(p) \sqcup \hll(X\setminus \Var(P'), p)\]
  gives a closure partition of $X\setminus \Var(P)$. 
  Denote $X_p:=X\setminus \Var(p)$ and $X_{P'}:=\hll(X\setminus \Var(P'), p)$. 
  For this, note that
  \[X_{P'} = \left[X\setminus \Var(P')\right]\setminus \zclo{X_p}.\] Hence we have
  $X\setminus \Var(P) \zeq X_p \sqcup X_{P'}$ and $X_p$ and $X_{P'}$ do not share
  any irreducible components by \Cref{prop:loccloideal} and
  \Cref{prop:hullcomps}.

  Now we just have to prove that \algname{Remove} terminates. Note
  that \Cref{line:crl2} is only called if $I_X\subsetneq I_Y$. Hence, by
  Noetherianity of $R$, during any run of \algname{Remove}, we can
  only call \Cref{line:crl2} finitely many times.  But then
  termination follows, since the input set $P$ is finite.
\end{proof}

Now we can give our main algorithm:

\begin{algorithm}[H]
  \caption{Computing an irredundant Kalkbrener partition}
  \label{alg:kalkdec}
  \begin{algorithmic}[1]
    \Require An affine cell $X$, an upper bound $c_X$ on the codimension of all irreducible components of $X$
    \Ensure An irredundant Kalkbrener partition of the locally closed set defined by $X$
    \Procedure{KalkPart}{$X$}
    \If{$\Call{Codim}{X} = c_X$}
    \Return \label{line:codimup} $\set{X}$
    \EndIf
    \State \label{line:foundsyz} $g,i\gets $\Call{GetSyz}{$X$}
    \If{$g=0$}
    \Return \label{line:lci} $\set{X}$
    \Else
    \State $X_1\gets \Call{AddInequation}{X, g}$
    \State Remove $f_i$ from the underlying sequence of $X_1$
    \State \label{line:r} $r\gets $number of elements in the underlying sequence of $X_1$
    \State \label{line:newcodimup} $c_{X_1} \gets \min\set{c_X,r}$
    \State $\mathcal{X}_2\gets $\Call{hull}{$X,g$}
    \For{$X_2$ in $\mathcal{X}_2$}
    \State \label{line:insertg} Insert $g$ into the underlying sequence of every $X_2\in \mathcal{X}_2$
    \EndFor
    \Return $\Call{KalkPart}{X_1,c_{X_1}}\cup \bigcup_{X_2\in \mathcal{X}_2}\Call{KalkPart}{X_2,c_X}$
    \EndIf
    \EndProcedure
  \end{algorithmic}
\end{algorithm}

Following from \Cref{thm:hllremove} we prove the correctness and termination
of \Cref{alg:kalkdec}:

\begin{corollary}
  The procedure \algname{KalkPart} terminates and is correct in that
  it satisfies its output specification.
\end{corollary}
\begin{proof}
  To establish correctness, let us first look at
  \Cref{line:codimup}. By definition, $c_X$ is an upper bound on the
  maximum codimension of all components of $X$. So, if
  $\codim(X) = c_X$, then $X$ is equidimensional of codimension
  $c$. To fully establish the correctness of this line, we now have to
  show, that, in \Cref{line:newcodimup}, $c_{X_1}$ is really an upper
  bound on the maximum codimension of all components of $X_1$. The
  closures of the components of $X_1$ are a subset of the closures of
  the components of $X$, so the maximum codimension of all components
  of $X_1$ is certainly bounded by $c_X$. With $r$ defined as in
  \Cref{line:r}, the codimension of all components of $X_1$ is bounded
  by $r$ by Krull's principal ideal theorem \parencite[Theorem 10.2
  in][]{eisenbud1995}. This shows that $c_{X_1}$ is really an upper
  bound on the codimension of all components of $X_1$.

  Next, in line \Cref{line:lci}, the equidimensionality of $X$ follows
  from \Cref{thm:conorm} and \Cref{prop:syz}.

  Finally note that $\set{X_1}\cup \mathcal{X}_2$, as defined in the algorithm,
  form a closure partition of $X$ by \Cref{thm:hllremove}. This is
  not affected by \Cref{line:insertg} because we have $g\in \rad{I_{X_2}}$
  for all $X_2\in \mathcal{X}_2$. This finally establishes the correctness of the
  algorithm.

  For termination, note that $I_X$ is strictly contained in the ideal
  $I_{X_1}$ and the ideals $I_{X_2}$ for $X_2\in \mathcal{X}_2$ if $g$ vanishes on
  some components of $X_1$: In this case $X_1$ and every
  $X_2\in \mathcal{X}_2$ has as irreducible components a strict subset of those of
  $X$ by \Cref{prop:loccloideal} and \Cref{prop:hullcomps}. So by
  Noetherianity of $R$, after a finite number of steps, every $g$
  found in \Cref{line:foundsyz} vanishes on no irreducible component
  of $X$. In this case, $\mathcal{X}_2 = \emptyset$ and
  $I_X = I_{X_1}$ but the sequence underlying the affine cell $X_1$ is
  strictly shorter than the one underlying $X$. This establishes the
  termination of \algname{KalkPart}.
\end{proof}

\section{Benchmarks}
\label{sec:bench}

In this section we show and discuss several benchmarks for
\Cref{alg:kalkdec}, comparing it with the algorithm for
equidimensional decomposition given in \textcite{eder2023a}.

For this we wrote an implementation of \Cref{alg:kalkdec} using the
\julia-package \AS which itself gives an interface to the Gröbner
basis library \msolve, itself written in \texttt{C}. Our
implementation uses this interface to \msolve to compute the
underlying Gröbner bases of all appearing affine cells and thus to
perform all operations given in \Cref{alg:basop}. As outlined in
\Cref{sec:datstruc}, we use our own implementation of signature-based
Gröbner basis (sGB) algorithms in \julia to facilitate the necessary
syzygy computations for the \algname{GetSyz} operation.

In such a sGB algorithm, the syzygies of a finite sequence of
polynomials $F$ correspond to {\em reductions to zero} in running this
algorithm on $F$, i.e. in computing a Gröbner basis for $\idl{F}$ such
an algorithm reduces certain polynomials by the current intermediate
Gröbner basis for $\idl{F}$ and if this reduction yields zero it
corresponds to a syzygy. From this reduction to zero, we may then
reconstruct the corresponding syzygy using the algorithm called
``Representation'' on p.8 in \textcite{sun2011}. If $F$ underlies an
affine cell $\Var(F,h,G)$ we then check membership of the entries of
the thusly recovered syzygy in $\idl{G}$ using normal form
computations.

The source code of our implementation is available at

\begin{center}
  \url{https://github.com/RafaelDavidMohr/AlgebraicSolving.jl}.
\end{center}

In \Cref{table:equidim2} we compare this implementation of \Cref{alg:kalkdec}
with several other methods:
\begin{itemize}
\item \primedec: The algorithm in \textcite{eder2023a}.
\item \elim: The homological method by \textcite{eisenbud1992}.
\item \homolog: The elimination method detailed in \textcite{decker1999}.
\item \regchain: Kalkbrener's algorithm as part of the Regular chains
  library in \sftware{Maple} (see \cite{lemaire2005}).
\end{itemize}

All examples were ran in a random prime characteristic with less than
$32$ bits on a single thread of an Intel Xeon Gold 6354. For each
example, we gave each used implementation at least an hour or fifty
times the time that the fastest implementation took, whichever is
bigger. If this time limit was exceeded, we indicated it with
$\infty$ in \Cref{table:equidim2}.

In order to make the four methods using Gröbner bases
(i.e. \Cref{alg:kalkdec}, \primedec, \elim and \homolog) as comparable
as possible, independent of implementational considerations, we tried
to use \msolve as the core Gröbner basis engine of all of them. Our
implementation accompanying the algorithm in \textcite{eder2023a} is
already based on \msolve. In addition, we implemented \elim using \AS,
this implementation can be found on the authors website at

\begin{center}
  \url{https://polsys.lip6.fr/~mohr/elim_decomp.jl}.
\end{center}

For \homolog, we used the \julia-based computer algebra system
\sftware{Oscar} \parencite{OSCAR-book} but supplied at least the
initial Gröbner basis of the input ideal using \msolve.

The polynomial systems on which we compared these algorithms are comprised
as follows:

\begin{itemize}
\item C1 and C3 are certain jacobian ideals of single multivariate
  polynomials which define singular hypersurfaces.
\item Cyclic$(8)$, encoding the standard cyclic benchmark in computer
  algebra in $8$ variables.
\item ED($d,n$) encodes the parametric \emph{euclidean distance}
  problem for a hypersurface of degree $d$ in $n$ variables,
  see~\textcite{draisma2014}.
\item PS$(n)$ encodes the singular points of an algebraic set cut
  out by polynomials
  $$f_1,\dots,f_{n-1},g_1,\dots,g_{n-1}$$
  with $f_i\in \mathbb{K}[x_1, \ldots, x_{n-2}, z_1, z_2]$,
  $g_i\in \mathbb{K}[y_1, \ldots, y_{n-2}, z_1, z_2]$, the $f_i$ being chosen
  as random dense quadrics, and $g_i$ chosen such that
  $g_i(x_1,\dots,x_{n-2},z_1,z_2) = f$, i.e. as a copy of $f_i$ in the
  variables $y_1,\dots,y_{n-2},z_1,z_2$.
\item Sing$(n)$ encodes the critical points of the restriction of the
  projection on the first coordinate to a (generically singular)
  hypersurface which is defined by the resultant in $x_{n+1}$ of two
\item SOS$(s, n)$ encodes the critical points of the restriction of
  the projection on the first coordinate to a hypersurface which is a
  sum of $s$ random dense quadrics in $\mathbb{K}[x_1, \ldots, x_n]$,
  random dense quadrics $A, B$ in $\mathbb{K}[x_1, \ldots, x_{n+1}]$.
\item The Steiner polynomial system, coming from
  \textcite{breiding2020}.
\item The two remaining examples sys2353 and sys2161 are part of the
  BPAS library \parencite{asadi2021} for triangular decomposition.
\end{itemize}

Let us close with a few remarks on the timings recorded in
\Cref{table:equidim2}. On most examples given in \Cref{alg:kalkdec},
\Cref{alg:kalkdec} is the most well behaved with a few exceptions:

On Cyclic(8) the bottleneck for \Cref{alg:kalkdec} lies entirely in
the required syzygy computations. For this example a significant
number of syzygies are computed and then verified to already lie in
the ideal underlying the affine cell in question, they therefore
cannot be used by \Cref{alg:kalkdec} to further decompose the affine
cell.

Despite it usually bringing a performance boost, processing all
defining equations at the same time can also be a bottleneck: In
particular at the start of the computation, a Gröbner basis of the
ideal defined by the input equations is required. This slows down
\Cref{alg:kalkdec} compared to \primedec in those cases where
\primedec performs a decomposition after only a small handful of the
input equations have been considered, this is the case for the
examples sys2353 and sys2161. These examples have many underlying
variables but each defining equation involves only a subset of these
variables, we observed that on such examples \regchain is also very
well behaved and refer to \textcite{eder2023a} for further benchmarks
indicating this.

Finally, we want to note that the examples ED(3,4) and ED(3,5) are
already equidimensional but not given by local regular
sequences. Therefore \Cref{alg:kalkdec} performs several decomposition
steps, while the equidimensionality of these examples is quickly
verified by \elim, which in this case just requires a single Gröbner
basis computation.

\begin{table}[hbt!]
  \caption{Comparing \Cref{alg:kalkdec} with other Methods for Equidimensional Decomposition}
  \label{table:equidim2}
  \begin{center}
    \small
    \begin{tabular}{llllll}
      \toprule
      & \Cref{alg:kalkdec} & \primedec & \elim & \homolog & \regchain\\
      \midrule
      Cyclic(8) & $\infty$ & 381.2 & $\infty$ & $\infty$ & $\infty$\\
      C1 & 15.0 & 129.0 & $\infty$ & $\infty$ & $\infty$\\
      C3 & 0.8 & 10.0 & $\infty$ & $7.2$ & $\infty$\\
      PS(10) & 0.2 & 1.7 & 5.1 & 32.1 & $\infty$\\
      PS(12) & 6.3 & 51.0 & 187.6 & $\infty$ & $\infty$\\
      PS(14) & 248.2 & 3128.3 & $\infty$ & $\infty$ & $\infty$\\
      PS(16) & 13666.2 & $\infty$ & $\infty$ & $\infty$ & $\infty$\\
      Sing(10) & 3.8 & 1.0 & $\infty$ & $\infty$ & $\infty$\\
      SOS(6,4) & 6.0 & 5.0 & $\infty$ & $\infty$ & $\infty$\\
      SOS(6,5) & 24.9 & 14.0 & $\infty$ & $\infty$ & $\infty$\\
      SOS(7,4) & 14.2 & 24.4 & $\infty$ & $\infty$ & $\infty$\\
      SOS(7,5) & 112.2 & $\infty$ & $\infty$ & $\infty$ & $\infty$\\
      Steiner & 404.9 & 870.0 & $\infty$ & $\infty$ & $\infty$\\
      sys2353 & 10.9 & 1.6 & 803.2 & 2.4 & 4.4\\
      sys2161 & 26.2 & 7.54 & $\infty$ & $\infty$ & 26.5\\
      ED(3,4) & 30.7 & 294.1 & 13.0 & $\infty$ & $\infty$\\
      ED(3,5) & 828.1 & $\infty$ & 283.5 & $\infty$ & $\infty$\\
      \bottomrule
    \end{tabular}
  \end{center}
\end{table}

\newpage
\printbibliography
\end{document}